\documentclass[12pt]{article}
\usepackage{amsmath,amssymb,amsfonts,amsthm}
\usepackage{graphicx}
\usepackage{color}
\usepackage{enumerate}
\newtheorem{theorem}{Theorem}[section]
\newtheorem{proposition}[theorem]{Proposition}

\newtheorem{lemma}[theorem]{Lemma}
\newtheorem{remark}{Remark}[section]
\newtheorem{example}{Example}[section]
\newtheorem{definition}{Definition}[section]
\newtheorem{claim}{Claim}

\definecolor{myred}{RGB}{255,50,50}         

\title{
Egalitarian solution for games \\ with discrete side payment}

\author{Takafumi Otsuka%
\thanks{Department of Economics and Business Administration,
Tokyo Metropolitan University, 
Tokyo 192-0397, Japan, 
\texttt{otsuka-takafumi1@ed.tmu.ac.jp}}
}
\date{\today}
\begin{document}
\maketitle

\begin{abstract}
In this paper, we study the egalitarian solution for games with discrete 
side payment,
where the characteristic function is integer-valued and payoffs of 
players are integral vectors.
The egalitarian solution,
introduced by Dutta and Ray in 1989,
is a solution concept for transferable utility cooperative games
in characteristic form,
which combines commitment for egalitarianism and promotion of indivisual 
interests in a consistent manner.
We first point out that the nice properties of the egalitarian solution 
(in the continuous case)
do not extend to games with discrete side payment.
Then we show that the Lorenz stable set, which may be regarded a variant of
the egalitarian solution, has nice properties
such as the Davis and Maschler reduced game property and the converse 
reduced game property.
For the proofs we utilize recent results
in discrete convex analysis on decreasing minimization on an M-convex set
investigated by Frank and Murota. 
\end{abstract}
\section{Introduction}
The egalitarian solution is a solution concept for transferable utility 
cooperative games
in characteristic form
which combines commitment for egalitarianism
and promotion of indivisual interests in a consistent manner.
This concept was introduced by Dutta--Ray (1989) \cite{Dutta-Ray1}.  
\par

The egalitarian solution is studied extensively 
in the literature.
For example, Arin et al. \cite{Arin_et_al}, Dutta \cite{Dutta}, and 
Klijn et al. \cite{Klijn}
axiomatise the egalitarian solution of Dutta--Ray for convex games.
Branzei et al. \cite{Branzei}, Dietzenbacher et al. 
\cite{Dietzenbacher}, Hokari \cite{Hokari,Hokari2},
and Llerena--Mauri \cite{Llerena=Mauri2017}
considered modifications
of the egalitarian solution
so that the solution
exists for a wider class of games.
These studies are mostly concerned with the case where the 
characteristic function is real-valued
and, accordingly, the side payment is real-valued.

Substantial connection has been recognized between the egalitarian solution
in convex games and the polymatroid theory in optimization.
Indeed, the core of a convex game is nothing but the base polyhedron of 
a polymatroid.
In the theory of polymatroids and submodular functions,
Fujishige (1980) \cite{Fujishige}
had introduced the concept of
{\em lexicographically optimal base},
and this concept is essentially equivalent to the egalitarian solution 
in convex games,
as noted by Fujishige \cite{FujishigeSubm} and Hokari--Uchida 
\cite{HokariUchida2004}.
In particular, the {\em principal partition} of Fujishige 
\cite{Fujishige} plays the decisive role
to clarify the properties of
the egalitarian solution in convex games
and also to develop algorithms for finding it.
\par

In this paper we are interested in the egalitarian solution for games 
with discrete side payment,
where the characteristic function is integer-valued and payoffs of 
players are integral vectors.
This study is partly motivated by a recent development in discrete 
convex analysis,
which is a theory of discrete convexity for functions on integer lattice 
points
(see \cite{Mdca98,Mdcasiam,Mbonn09,Mdcaeco16}).
Frank--Murota \cite{FK1,FK2}
recently investigated  the discrete decreasing minimization problem,
which is concerned with lexicographically minimal integral vectors in an 
integral base polyhedron.
This is a discrete counterpart of the work by Fujishige \cite{Fujishige},
and, in particular, the discrete counterpart of the principal partition
is established as the {\em canonical partition}.
The main objective of this paper is to clarify the properties of the 
egalitarian solution
in games with discrete side payment
by making use of these results of lexicographically minimal 
(decreasingly minimal)
integral elements in an integral base polyhedron. 
\par

The results of this paper are summarized as follows.
First, we show by an example that, unlike the case of $\mathbf{R}$,
the egalitarian solution for convex games with discrete side payment is 
not equivalent to
the lexicographically minimal (decreasingly minimal) element.
Accordingly, the egalitarian solution in the case of $\mathbf{Z}$ 
fails to have nice properties
of the egalitarian solution in the continuous case.
This motivates us consider the Lorenz stable set 
(or equivalently Lorenz maximal imputation).
The Lorenz stable set is
introduced by Arin--Inarra \cite{Arin_Inarra} and Hougaard et al. 
\cite{HougaardPelegPetersen2001}
and is defined as the subset of the core consisting of the elements
that are not Lorenz-dominated by any other element of the core.
We show that the Lorenz stable set has nice properties
such as the Davis and Maschler reduced game property \cite{DM} and the 
converse reduced game property \cite{Peleg}.
Our analysis of the Lorenz stable set relies heavily on the
recent results on the discrete decreasing minimization problem.
\par

This paper is organized as follows.
Sections \ref{Section:EgainR} and \ref{Section:PolyandDDM} 
are brief reviews on the egalitarian solution
and discrete decreasing minimization problem, respectively.
In Section \ref{Section:EgainZ}, 
we investigate the properties of the egalitarian solution
in games with discrete side payment
in comparison with the egalitarian solution in continuous variables.
In Sections \ref{Section:RGP} and \ref{Section:CRGP}, 
we clarify the fundamental properties of the Lorenz 
stable set
by utilizing the results
on discrete decreasing minimization.
\section{Egalitarian solution in the continuous case}
\label{Section:EgainR}
We provide a brief summary on the egalitarian solution 
of Dutta--Ray \cite{Dutta-Ray1}.

\subsection{Definition and Notation}\label{Subsec2Def}
We consider a transferable utility game 
in characteristic function form. 
There are $n$ players and let $N = \{1, 2, \dots, n\}$. 
A {\em coalition} is a nonempty subset of $N$, 
whereas
$N$ is called the {\em grand coalition}. 
The worth of a coalition $S$ is given by a scalar $v(S)$. 
We assume $v(\emptyset) = 0$ throughout this paper.
A pair $(N,v)$ is called a {\em game}.
We denote the set of games by $\Gamma$. 
\par

For a vector $x \in \mathbf{R}^N$, 
we sometimes abbreviate $\sum_{i = 1}^N x_i$ to $x(N)$.
Let $(N,v)$ be a game. The set 
$X(N,v) = \{x \in \mathbf{R}_+^N \mid x(N) \leq v(N)\}$
is called the set of {\em feasible payoff vectors} 
for the game $(N,v)$. 
A {\em solution} on $\Gamma$ is a function $\sigma$ 
which associates with each game $(N,v) \in \Gamma$ 
a subset $\sigma(N,v)$ of $X(N,v)$.
\par

For a game $(N,v)$, we call $x$ an {\em imputation} 
if $x_i \geq v(\{i\})$ for all $i \in N$ and $x(N) = v(N)$. 
The restriction of $N$ to $S \subseteq N$ is denoted by $x_S$. 
For two vectors $x, y \in \mathbf{R}^N$, 
we write $x > y$ if $x_i \geq y_i$ 
for all $i = 1, \dots, n$, with strict inequality for some $i$. 
For a vector $x$, let $x{\downarrow}$ denote 
the vector obtained from $x$ by rearranging its components 
in a decreasing order. 
The Lorenz-domination is defined as follows:

\begin{definition}\label{Def:LorenzdominateDec}
\rm
For two vectors $x$ and $y$ 
in $\mathbf{R}^N$ with $x(N) = y(N)$, 
we say that $x$ {\em Lorenz-dominates} $y$ 
if $\sum_{j=1}^i (x{\downarrow})_j \leq \sum_{j=1}^i (y{\downarrow})_j$ 
holds for all $i = 1, \dots, n$, 
with strict inequality for some $i$.
\qed\end{definition}

Also, for $x, y \in \mathbf{R}^N$, 
we say that $x$ and $y$ are {\em value-equivalent}
if $x{\downarrow} = y{\downarrow}$ holds.

Throughout this paper, we define the Lorenz-domination
by a decreasing order in accordance with Dutta--Ray \cite{Dutta-Ray1}. 
Other papers, however, adopt an increasing order as follows:

\begin{definition}\label{DefofLorenzdominateInc}
\rm
For a vector $x$, 
let $x{\uparrow}$ denote the vector obtained from $x$ 
by rearranging its components in an increasing order. 
For two vectors $x$ and $y$ in $\mathbf{R}^N$ with $x(N) = y(N)$, 
we say that $x$ {\em Lorenz-dominates} $y$ 
if $\sum_{j=1}^i (x{\uparrow})_j \geq \sum_{j=1}^i (y{\uparrow})_j$ 
for all $i = 1, \dots, n$, 
with strict inequality for some $i$. 
\end{definition}
The decreasing order is adopted in 
Dutta \cite{Dutta}, 
Dutta--Ray \cite{Dutta-Ray1}, 
Llerena \cite{Llerena2012}, and
Llerena--Mauri \cite{Llerena=Mauri2016}, whereas 
the increasing order is in
Arin--Inarra \cite{Arin_Inarra}, 
Arin et al. \cite{AKV2008}, 
Fei--Fields \cite{FieldsFei1978}, 
Hokari \cite{Hokari, Hokari2},  
Llerena--Mauri \cite{Llerena=Mauri2017}, and
Shaked--Shanthikumar \cite{SS2007}.
We note that 
this difference in the definition of Lorenz-domination 
does not affect the results of this paper; 
see Remark \ref{equivaofIncmaxDecmin}. 
We also mention that the Lorenz-domination is defined equivalently as follows. 
For two payoff vectors $x, y \in \mathbf{R}^N$
with $x(N) = y(N)$, 
\begin{align}
x ~ \mbox{Lorenz-dominates} ~ y \Leftrightarrow \mbox{Its Lorenz curve lies nowhere below that of} ~ y. \nonumber
\end{align}
\noindent
For example, Patrick \cite{Moyes1987} and 
Tatiana \cite{TatianaDamjanovic} adopt this definition. 
\par

Next, we define the core, 
which is a central solution concept 
of cooperative game theory 
(cf., \cite{PelegSudholter}).

\begin{definition} \label{defofcore}
\rm
For any coalition $S$, the {\em core} of $S$ is defined by
\[
\pushQED{\qed}
C(S,v) = \{ x \in \mathbf{R}^S \mid 
x(S) = v(S), x(T) \geq v(T) ~ (\forall T \subsetneq S) \}. \qedhere
\popQED
\]
\end{definition}

\begin{definition}\label{Defofconvexgame}
\rm
We call a game $(N,v)$ a
{\em convex game}
if $v$ is a supermodular function, that is,
\begin{equation*}
v(S) + v(T) \leq v(S \cup T) + v(S \cap T)  
\end{equation*}
for all $S, T \subseteq N$.
We denote the set of convex games by $\Gamma^{\rm c}$.
\qed\end{definition}

The notion of the egalitarian solution will now be described. 
First, 
the {\em Lorenz map} $E$ is defined on the domain 
$\{ A \mid A\subseteq \mathbf{R}^k  \,(\exists k \in \{1, 2, \dots, n\}), 
\exists u \in \mathbf{R},  \forall x \in A \colon 
\sum_{i=1}^k x_i = u \}$. 
For each such set $A$, 
$E(A)$ denotes the set of all elements in $A$ that 
are not Lorenz-dominated within $A$. 
\par
	
Next, we define the {\em Lorenz core} 
introduced by Dutta--Ray \cite{Dutta-Ray1}.
See \cite{Dutta-Ray1} for the details about the Lorenz core. 
The Lorenz core is defined recursively as follows. 
The Lorenz core of a singleton coalition is 
$L(\{i\}) = \{v(\{i\})\} ~ (i \in N)$. 
Now suppose that the Lorenz cores 
for all coalitions of cardinality $k-1$ or less have been defined, 
where $2 \leq k < n$. 
The Lorenz core of coalitions of size $k$ is defined by
\begin{align*}
L(S,v) = \{x \in \mathbf{R}^S \mid x(S) = v(S)~\mbox{and there is no} ~ T \subsetneq S~\mbox{and} \\ 
y \in E(L(T,v))~\mbox{such that~ }x_T < y \}.
\end{align*}
\noindent
For a game $(N,v) \in \Gamma$, 
we call an element of $E(L(N,v))$ an {\em egalitarian solution}.
There is an inclusion between the core and the Lorenz core.

\begin{proposition}\label{Prop:InclusionCoreLo}
For any $S \subseteq N$, 
\begin{align}
C(S,v) \subseteq L(S,v).  \label{eq:coreLorenzcore}
\end{align}
\end{proposition}
\begin{proof}
Assume that
$x \notin L(S,v)$, 
which implies that 
$y > x_T$
for some $T \subsetneq S ~ (T \neq \emptyset)$ 
and some $y \in E(L(T,v))$.
By this inequality and $y(T) = v(T)$, 
we obtain $x(T) < y(T) = v(T)$, 
which implies $x \notin C(S,v)$.
Therefore, (\ref{eq:coreLorenzcore}) holds.
\end{proof}

\begin{remark}\label{Remark:WeightedEga}
\rm
For the weighted egalitarian solution, 
see Hokari \cite{Hokari2} and Koster \cite{Koster1999}.
\end{remark}

\subsection{Properties of egalitarian solution}\label{Sub2:PropertyofEga}
Here we describe the properties of the egalitarian solution 
shown by Dutta--Ray \cite{Dutta-Ray1}.
First, the following theorem shows that 
the egalitarian solution is unique if it exists at all.

\begin{theorem}\label{Theorem:ExistenceofEgaR}{\rm (Dutta--Ray \cite{Dutta-Ray1})}
There is at most one egalitarian solution in any game.
\qed\end{theorem}

Note that Theorem \ref{Theorem:ExistenceofEgaR} does not guarantee 
the existence of the egalitarian solution. 
The next theorem reveals that in any convex game, 
the egalitarian solution always exists and belongs to the core.

\begin{theorem}\label{Theorem:RelationCoreEga}{\rm (Dutta--Ray \cite{Dutta-Ray1})}
In convex games, 
an egalitarian solution exists and 
it is contained in the core, 
that is, 
for any $(N, v) \in \Gamma^{\rm c}$, 
we have $\emptyset \neq E(L(N,v)) \subseteq C(N,v)$.
\qed\end{theorem}

Moreover, the egalitarian solution has a nice property as follows.

\begin{theorem}\label{Theorem:EgaLorenzdominance}{\rm (Dutta--Ray \cite{Dutta-Ray1})}
In convex games, 
the egalitarian solution Lorenz-dominates 
every other element of the core.
\qed\end{theorem}

Theorem \ref{Theorem:EgaLorenzdominance} raises the question 
whether the egalitarian solution Lorenz-dominates 
every other element of the Lorenz core. 
The following example shows that 
this is not true even in convex games.

\begin{example}\label{Example:Lorenzcore}\rm (Dutta--Ray \cite{Dutta-Ray1}, Example 5)
Let $N = \{1,2,3\}$, $v(\{1\}) = 4$，$v(\{2\}) = 6$，
$v(\{3\}) = 8$，$v(\{1,2\}) = 11$，$v(\{1,3\})  =12$，
$v(\{2,3\}) = 15$ and $v(N) = 21$. 
This game is convex. 
The egalitarian solution in this game is $(6,7,8)$, 
which does not Lorenz-dominate $(6.25, 6.5, 8.25) \in L(N,v)$. 
\qed\end{example}

\subsection{The algorithm for egalitarian solution in a convex game}
\label{Sub2:decompositonalgorithm}
We describe an algorithm 
to locate the egalitarian solution 
introduced by  Dutta--Ray \cite{Dutta-Ray1}. 
This algorithm is equivalent to the {\em decomposition algorithm} 
of Fujishige \cite{Fujishige}.

\medskip
Let $v: 2^N \rightarrow \mathbf{R}$ be a supermodular set function. 
Define $v_1 = v$.\par
\noindent
Step 1: Let $S_1$ be the coalition 
that satisfies the following two conditions. \par
\begin{enumerate}
\item $v_1(S_1) / |S_1| \geq v_1(S) / |S| $ 
for all $S \subseteq N$. \par
\item $|S_1| > |S|$ for all $S \neq S_1$ 
such that $v_1(S_1) / |S_1| = v_1(S) / |S|$.
\end{enumerate}
\noindent
That is, 
$S_1$ is the largest coalition having the highest average worth. 
By using supermodularity of $v$, 
we can verify the existence of such an $S_1$. 
Define
\begin{equation*}
x^*_i = \frac {v_1(S_1)}{|S_1|} \quad (i \in S_1).
\end{equation*}
\noindent
Step $k \, (k \geq 2)$: 
Suppose that $(S_1, v_1), \dots, (S_{k-1}, v_{k-1})~ (k \geq 2)$
have been defined 
and $S_1 \cup \dots \cup S_{k-1} \neq N$. 
Define a new game with player set 
$N \setminus ({S_1 \cup \dotsb \cup S_{k-1}})$. 
For all coalitions $S$ of this new player set, 
define $v_k(S)$ by
\begin{equation*}
v_k(S) = v_{k-1}(S_{k-1} \cup S) - v_{k-1}(S_{k-1}).
\end{equation*}
\noindent

By the definition of $v_k$, 
$v_k$ is a supermodular function.
Just as in Step 1, 
define $S_k$ to be the largest coalition in 
$N \setminus (S_1 \cup \dots \cup S_{k-1})$ 
that maximizes $\frac{v_k(S)}{|S|}$
and define 
\begin{equation*}
x^*_i = \frac{v_k(S_k)}{|S_k|} \quad (i \in S_k).
\end{equation*}

In at most $n$ steps, 
we can obtain a partition of $N$ into sets $S_1, \dots, S_m ~ (m \leq n)$ and 
the egalitarian solution. 
By the above construction of $x^*$, 
we obtain the following:
\begin{equation}
x^*_i = x^*_j \quad ( i, j \in S_l, ~ l = 1, \dots, m), \label{eq:uniform}
\end{equation}
\begin{equation}
\sum_{k=1}^l \sum_{j \in S_k} x^*_j = v(S_1 \cup \dots \cup S_l) \quad (l = 1, \dots, m),
\end{equation}
\begin{equation}
x^*_i > x^*_j \quad (i \in S_k,~ j \in S_l, ~ k < l ). \label{eq:different}
\end{equation}
\noindent
The equation (\ref{eq:uniform}) shows that 
for $l = 1, \dots, m$, 
each payoff of the players belonging to $S_l$ is the same.   
\section{Polymatroid theory and decreasing minimization problem}
\label{Section:PolyandDDM}
In this section, we overview the results of
the polymatroid theory and decreasing minimization problem
from discrete convex analysis.

\subsection{Definition and Notation}\label{Sub3:Def}
First, we give the basic facts about majorization
and decreasing minimization.
A vector $x$ is decreasingly smaller than vector $y$, 
in notation $x <_{\rm dec} y$, 
if $x{\downarrow}$ is lexicographically smaller than $y{\downarrow}$ 
in the sense that they are not value-equivalent and 
$(x{\downarrow})_j < (y{\downarrow})_j$ 
for the smallest subscript $j$ for which $(x{\downarrow})_j$ 
and $(y{\downarrow})_j$ differ. 
We write $x \leq_{\rm dec} y$ to mean 
that $x$ is decreasingly smaller or value-equivalent to $y$. 
For a set $Q$ of vectors, 
$x \in Q$ is {\em decreasingly minimal} 
(dec-min, for short) in $Q$
if $x \leq_{\rm dec} y$ holds for every $y \in Q$. 
\par

The {\em decreasing minimization problem} is 
to find a dec-min element of a given set $Q$ of vectors.
Frank--Murota \cite{FK1,FK2} deal with the case 
where the set $Q$ is an M-convex set, 
which is to be defined in Section \ref{Sub3:Mconvex}. 
\par

Just as the notion of decreasingly minimality, 
we can consider a notion of increasingly maximality. 
A vector $x$ is increasingly larger than vector $y$, 
in notation $x >_{\rm inc} y$, 
if $x{\uparrow}$ is lexicographically larger than $y{\uparrow}$ 
in the sense that they are not value-equivalent 
and $(x{\uparrow})_j > (y{\uparrow})_j$ 
for the smallest subscript $j$ for which $(x{\uparrow})_j$ 
and $(y{\uparrow})_j$ differ. 
We write $x \geq_{\rm inc} y$ 
to mean that $x$ is increasingly larger 
or value-equivalent to $y$.
For a set $Q$ of vectors, 
$x \in Q$ is {\em increasingly maximal} 
(inc-max, for short) in $Q$
if $x \geq_{\rm inc} y$ holds for every $y \in Q$. 
\par

Let $\overline{x}$ denote the vector 
whose $k$-th component $\overline{x}_k$ 
is equal to the sum of the first $k$ components of $x{\downarrow}$. 
A vector $x$ is said to be {\em majorized} by another vector $y$, 
in notation $x \prec y$, 
if $\overline{x} \leq \overline{y}$ 
and $\overline{x}_n = \overline{y}_n$ hold. 
Also, $x$ is said to be {\em strictly majorized} by $y$ 
if $\overline{x} < \overline{y}$ 
and $\overline{x}_n = \overline{y}_n$ hold \cite{MOA11}. 
Let $Q$ be an arbitrary subset of $\mathbf{R}^N$. 
An element $x$ of $Q$ is said to be {\em least majorized} in $Q$ 
if $x$ is majorized by all $y \in Q$. 
\par

There exists a relationship 
between the notion of decreasing minimality and 
that of being least majorized as follows.

\begin{proposition}\label{PropofTamir}{\rm (e.g., Frank--Murota \cite{FK2} and Tamir \cite{Tamir})}
Let $Q$ be an arbitrary subset of $\mathbf{R}^N$ 
and assume that $Q$ admits a least majorized element. 
Then, an element of Q is least majorized in Q 
if and only if 
it is decreasingly minimal in Q.
\qed\end{proposition}

Also, there exists a relationship 
between being majorized and Lorenz-domination, 
that is, 
\par
\begin{center}
$x$ Lorenz-dominates $y$ $\Leftrightarrow$ $x$ is strictly majorized by $y$.
\end{center}
\noindent

Note that if we replace ``strictly majorized'' with ``majorized'' 
in the above, 
then $\Rightarrow$ is true but $\Leftarrow$ is not true. 
Indeed, if $\overline{x}_j = \overline{y}_j$ for all $j = 1, \dots, n$, 
then $x$ is majorized by $y$ but $x$ does not Lorenz-dominate $y$. 
However, if we identify value-equivalent vectors, 
the notion of Lorenz-domination is equivalent to that of being majorized. 
\par

\subsection{Polymatroid theory and decreasing minimization on an M-convex set}
\label{Sub3:Mconvex}
In polymatroid theory, 
the concept of base polyhedron plays a central role.
A base polyhedron is defined as follows.

\begin{definition}\label{DefofBasepoly}
\rm
For a finite-valued supermodular set function $g$ on $N$
with $g(\emptyset) = 0$,
the  associated {\em base polyhedron} $B(g)$ is defined by
\[
\pushQED{\qed}
B(g) = \{x \in \mathbf{R}^N \mid x(N) = g(N), x(S) \geq g(S) ~ (\forall
S \subsetneq N)\}.  \qedhere
\popQED
\]
\end{definition}

If $g$ is an integer-valued supermodular set function, 
we call $B(g)$ an integral base polyhedron. 
Any extreme point of an integral base polyhedron is an integer point and 
the convex hull of the integer points of $B(g)$ 
coincides with $B(g)$ itself.

Lexicographically optimal base is defined as follows \cite{Fujishige}: 

\begin{definition} \label{DefofLexco}
\rm 
For a supermodular function $g \colon 2^N \rightarrow \mathbf{R}$, 
$x \in B(g)$ is said to be a {\em lexicographically optimal base} of $B(g)$ 
if $x \geq_{\rm inc} y$ holds for any $y \in B(g)$.
\qed\end{definition}

\begin{remark}\label{remark:FujiLexmax}
\rm
Fujishige \cite{Fujishige} deals with the weighted lexicographically optimal base. 
We mainly treat the unweighted lexicographically optimal base in this paper.
\qed\end{remark}

Next, we define an M-concex set, 
which plays a central role in discrete convex analysis.
Here for two vectors $x, y \in \mathbf{R}^N$, 
we define 
\begin{align*}
&{\rm supp}^+(x-y) = \{i \in N \mid x_i > y_i\},  \\
&{\rm supp}^-(x-y) = \{j \in N \mid x_j < y_j\}  
\end{align*}
\noindent
and for each $i \in N$, let characteristic vector $\chi_i \in \{0, 1\}^N$ denote by
\begin{align*}
\chi_i(j) = \begin{cases}
1 & (i = j),  \\
0& (i \neq j). 
\end{cases}
\end{align*}
\noindent
Then, the concept of M-convex set is defined as follows.

\begin{definition} \label{DefofMconcexset} \rm (M-convex set, Murota \cite{MdcaSteinitz's,Mdca98,Mdcasiam})
We say that the set $B \subseteq \mathbf{Z}^N$ is an {\em M-convex set}
if for any $x, y \in B$ 
and for any $i \in {\rm supp^+}(x-y)$, 
there exists some $j \in {\rm supp^-}(x-y)$:
\[
\pushQED{\qed}
x - \chi_i + \chi_j \in B,  \, y + \chi_i - \chi_j\in B. \qedhere
\popQED
\]
\end{definition}
A set $B \subseteq \mathbf{Z}^N$ 
is an M-convex set 
if and only if $B = B(g) \cap \mathbf{Z}^N$ 
holds for some integer-valued supermodular function $g$. 
That is, 
an M-convex set is nothing but 
the set of integral points of an integral base polyhedron.
\par

Decreasingly minimal elements on an M-convex set 
can be regarded as a discrete counterpart of the lexicographically optimal base.
We note that for any M-convex set $B$, 
an element is decreasingly minimal in $B$ 
if and only if 
it is increasingly maximal in $B$ (cf., \cite{FK2,Tamir}).
Frank--Murota \cite{FK1,FK2} mainly consider the problem of
finding a dec-min element of an M-convex set.
\par

M-convex set is characterized by the exchange axiom 
of Definition \ref{DefofMconcexset}.
Similarly, a dec-min element of an M-convex set is 
characterized by certain exchange operations as follows.

\begin{definition}\label{Defof1tight}\rm (1-tightening step \cite{FK1,FK2}) 
Let $B \subseteq \mathbf{Z}^N$ be an M-convex set.
A {\em 1-tightening step} for $x \in B$ means 
the operation of replacing $x$ to $x + \chi_i - \chi_j$ 
for some $i, j \in N$ such that 
$x_j \geq x_i + 2$ and $x + \chi_i - \chi_j \in B$.
\qed\end{definition}

\begin{theorem}\label{1tighteningDecminTheorem}
{\rm (Frank--Murota \cite{FK1}, Theorem 3.3)} 
Let $B$ be an M-convex set.
For an element $x$ of $B$, 
the following equivalence holds:
\par
\begin{center}
There is no 1-tightening step for $x$ $\Leftrightarrow$
$x$ is decreasingly minimal in $B$. 
\qed\end{center}
\end{theorem}

\begin{remark}\label{Remarkof1tight}
\rm
A 1-tightening step is called the {\em Robin Hood transfer} 
or {\em Robin Hood operation} 
in economics and the theory of majorization 
(see also Arnold \cite{Arnold1987} and Marshall et al. \cite{MOA11}). 
Also, a 1-tightening step is called the {\em progressive transfer} 
or {\em rich to poor transfer} in Dutta--Ray \cite{Dutta-Ray1}. 
Note that they do not restrict $B$ to an M-convex set 
in Definition \ref{Defof1tight}.
\qed\end{remark}

Finally, we explain the relationship between 
an M-convex set and a least majorized element.
The following theorem shows 
the existence of a least majorized element in an M-convex set.

\begin{theorem}\label{theorem:Mleastmajorized}{\rm (e.g., \cite{FK2,Tamir})} 
An M-convex set admits a least majorized element.
\qed\end{theorem}

\begin{remark}
\rm
According to Frank--Murota \cite{FK2}, 
the above fact has long been recognized by experts at least since 1995, 
though it was difficult to identify its origin in the literature. 
\qed\end{remark}

\subsection{Structure of dec-min elements on an M-convex set}\label{Subs:StructureofDecmin}
In this section, 
we introduce a partition and a chain called {\em canonical partition} 
and {\em canonical chain} respectively 
that describe the structure of dec-min elements on an M-convex set. 
They are introduced by Frank--Murota \cite{FK1,FK2} 
and the canonical partition is a discrete counterpart of the principal partition
considered by Fujishige \cite{Fujishige}
for the lexicographically optimal base
in continuous variables. 
\par

The canonical chain and the canonical partition of an M-convex set 
are constructed as follows \cite{FK1,FK2}. 
Let $g: 2^N \rightarrow \mathbf{Z}$ 
be an integer-valued supermodular function 
with $g(\emptyset) = 0$ and $g(N) > -\infty$. 
Consider the smallest maximizer $L(\beta)$ of $g(X) - \beta|X|$ 
for all integers $\beta$. 
There are finitely many $\beta$ for which $L(\beta) \neq L(\beta - 1)$.
Denote such integers as 
$\beta_1 > \beta_2 > \dots > \beta_q$ 
and call the {\em essential value-sequence}. 
Furthermore, 
define $C_k = L(\beta_k -1)$ for $k = 1, 2, \dots, q$ 
to obtain a chain: 
$C_1 \subsetneq C_2 \subsetneq C_2 \subsetneq \dots \subsetneq C_q$. 
Call this the {\em canonical chain}. 
Finally define a partition $\{S_1, S_2, \dots, S_q\}$ of $N$ 
by $S_k = C_k \setminus C_{k-1}$ for $k = 1, 2, \dots, q$, 
where $C_0 = \emptyset$, 
and call this the {\em canonical partition}. 
\par 

Alternatively, the canonical chain and the canonical partition 
can be defined iteratively as follows. 
For $k = 1, 2, \dots, q$, 
define 
\begin{align*}
&\beta_k = \max\biggl\{\biggl\lceil \frac{g(X \cup C_{k-1}) - g(C_{k-1})}{|X|}\biggr\rceil \mid \emptyset \neq X \subseteq \overline{C_{k-1}}\biggr\},   \\
&h_k(X) = g(X \cup C_{k-1}) - (\beta_k - 1)|X| - g(C_{k-1}) \quad (X \subseteq \overline{C_{k-1}}),   \\
&S_k = {\rm smallest ~ subset ~ of} ~ \overline{C_{k-1}}~ {\rm maximizing} ~ h_k, \\
&C_k = C_{k-1} \cup S_k, 
\end{align*}
\noindent
where $\overline{C_{k-1}} = N \setminus C_{k-1}$. 
\par

Then, this chain enables us to construct 
the set of dec-min elements on an M-convex set as follows.
We define the supermodular function 
$g'_k:2^{\overline{C_k}} \rightarrow \mathbf{Z}$ 
\begin{align*}
g'_k(X) = g(X \cup C_k) - g(C_k) \quad (X \subseteq \overline{C_k}) 
\end{align*}
\noindent
which defines the M-convex set $B'_k = B'(g'_k)$ in $\mathbf{R}^{\overline{C_k}}$.

Moreover, we denote the restriction of $g'_k$ to $S_k$ by $g_k$ 
which defines the M-convex set $B_k = B'(g_k) \subseteq \mathbf{R}^{S_k}$
for each $k = 1, \dots, q$.
Let $B^{\oplus}$ denote the face of $B(g)$ 
defined by the canonical chain 
$\mathcal{C}^* = \{C_1, \dots, C_q\}$, 
that is, $B^{\oplus}$ is the direct sum of the M-convex sets 
$B_k ~ (k = 1, \dots, q)$.
We note that $B^{\oplus}$ is an M-convex set
because the direct sum of M-convex sets is an M-convex set 
(cf., \cite{Mdcasiam}). 
\par

We define  $T_k ~ (k = 1, \dots, q)$ 
and $T^{\oplus}$, which is the direct sum of $T_k ~ (k = 1, \dots, q)$, 
by using the essential value-sequence as follows:
\begin{align*}
&T_k = \{x \in \mathbf{Z}^{S_k} \mid \beta_k - 1 \leq x_i \leq \beta_k ~ (i \in S_k)\},\\
&T^{\oplus} = \{x \in \mathbf{Z}^{N} \mid \beta_k -1 \leq x_i \leq \beta_k ~ (i \in S_k), k = 1,\dots, q\}.
\end{align*}
\noindent
The intersection of an M-convex set with an integral box 
is always an M-convex set,
and hence $B^{\oplus} \cap T^{\oplus}$ is an M-convex set.
The following theorem shows 
that the set of dec-min elements of an M-convex set forms an M-convex set
and is characterized by the canonical partition.

\begin{theorem}\label{Theorem:structureofDecmin}
{\rm (Frank--Murota \cite{FK1}, Theorem 5.1)} 
The set of decreasingly minimal elements of $B(g)$ is 
$B^{\oplus} \cap T^{\oplus}$.
That is, 
an element $x \in B(g)$ is decreasingly minimal in $B(g)$
if and only if 
$x_{S_k} \in B_k \cap T_k$ holds for each $k = 1, \dots, q$.
\qed\end{theorem}

This theorem also implies that 
for every dec-min element $x$ of $B(g)$, 
we have
\begin{align}
x_i \in \{\beta_k-1,\beta_k\} \quad (i \in S_k). \label{eq:near-uniform}
\end{align}

\section{Egalitarian solution in the discrete case}
\label{Section:EgainZ}
In this section, 
we investigate the properties of the egalitarian solution 
in Case $\mathbf{Z}$ 
in comparison with the case $\mathbf{R}$. 

\subsection{Preliminaries on the egalitarian solution in the discrete case}
\label{Sub4:DefofEgaZ}
We first define a game with discrete side payment. 
A {\em game with discrete side payment} will mean a game $(N,v)$, 
where the characteristic function $v$ is integer-valued 
and payoffs of players are integral vectors. 
We call a game with discrete side payment 
a {\em discrete game} for short
and denote the set of discrete games 
by $\Gamma_{\mathbf{Z}}$. 
Also, we say that a discrete game is 
a {\em discrete convex game} 
when its characteristic function is supermodular. 
We denote the set of discrete convex games 
by $\Gamma^{\rm c}_{\mathbf{Z}}$. 
For clarity, 
we denote the set of games in continuous variables 
and the set of convex games in continuous variables 
by $\Gamma_{\mathbf{R}}$ and $\Gamma^{\rm c}_{\mathbf{R}}$, 
respectively.
\par

We define the egalitarian solution in discrete games 
by simply replacing $\mathbf{R}$ with $\mathbf{Z}$ 
in the definitions of Section \ref{Section:EgainR}. 
Specifically, 
it is defined as follows.
For a discrete game $(N,v) \in \Gamma_{\mathbf{Z}}$, 
the Lorenz core of a singleton coalition is 
$L(\{i\}) = \{v(\{i\})\} ~ (i \in N)$. 
We note that since $v$ is integer-valued, 
$L(\{i\})$ is a set of an integral vector for each $i \in N$.
Now suppose that the Lorenz cores 
for all coalitions of cardinality $k-1$ or less have been defined, 
where $2 \leq k < n$. 
The Lorenz core of coalitions $S$ of size $k$ is defined by
\begin{align*}
L(S,v) = \{x \in \mathbf{Z}^S \mid x(S) = v(S)~\mbox{and there is no} ~ T \subsetneq S~\mbox{and} \\
y \in E(L(T,v))~\mbox{such that~ }x_T < y \}.
\end{align*}
\noindent
By the definition, 
$L(S,v)$ is composed of integral vectors for each $S \subseteq N$.
\par

Then, 
analogous to the Case $\mathbf{R}$,
we call an element of $E(L(N,v))$ 
an {\em egalitarian solution}
for the discrete game $(N,v) \in \Gamma_{\mathbf{Z}}$,
where $E(L(N,v))$ is a subset of the Lorenz core $L(N,v)$
that are not Lorenz-dominated 
by any other element of the Lorenz core. 
\par

The main properties of the egalitarian solution in Case $\mathbf{R}$ 
are as follows:

\begin{enumerate}[Property 1]
\item 
There is at most one egalitarian solution in any game. (Uniqueness)

\item 
In convex games, there exists an egalitarian solution and it is in the core.

\item 
In convex games, the egalitarian solution 
Lorenz-dominates every other element of the core.
\end{enumerate}

We will investigate 
whether the egalitarian solution in the discrete cases
has these properties. 
In this section, we show the following: 
 
\begin{itemize}
\item In discrete games, there may exist multiple egalitarian solutions 
(Example \ref{ExampleofMultipleEga}). 
\item 
In discrete convex games, 
there exists at least one egalitarian solution 
(Theorem \ref{ExistenceofEgalitarian}).
\item 
In discrete convex games, 
there may exist an egalitarian solution outside the core 
(Example \ref{CounterExample}). 
\item 
In discrete convex games, 
every element of egalitarian solutions in the core, 
if any, Lorenz-dominates every element of the core 
that is not an egalitarian solution 
(Theorem \ref{Existenceofleastmajorizedelement}). 
In addition, the egalitarian solutions outside the core 
do not necessarily Lorenz-dominate every element of the core 
that is not an egalitarian solution 
(Example \ref{CounterExample}).
\end{itemize}

\subsection{Relationship between the egalitarian solution and polymatroid theory}
\label{Sub4:EgaandPoly}
In this subsection, we describe the connection 
between the egalitarian solution in convex games 
and the polymatroid theory.
\par
First, it is obvious from the definitions that 
if $(N,v)$ is a convex game, 
then the core $C(N,v)$ coincides with the base polyhedron. 
Then, there is the following relationship 
between the lexicographically optimal base and the egalitarian solution.

\begin{theorem}\label{Theorem:LexEga}{\rm (\cite{FujishigeSubm,HokariUchida2004})} 
In convex games $\Gamma^{\rm c}_{\mathbf{R}}$, 
the lexicographically optimal base 
is equivalent to the egalitarian solution.	
\qed\end{theorem}

On the other hand, 
as is shown in Section \ref{PropertyofEgaZ}, 
the lexicographically optimal base (decreasingly minimal element) 
is not equivalent to the egalitarian solution in discrete convex games. 
\par

Next, we investigate the relationship 
between the core and an M-convex set. 
For a game $(N,v)$, 
if $v$ is an integer-valued supermodular function, 
then the core of the game is an integral base polyhedron. 
Since an M-convex set is the set of 
integral members of an integral base polyhedron, 
the core of a discrete convex game is an M-convex set. 
Therefore, the following property holds in discrete convex games 
by Theorem \ref{theorem:Mleastmajorized}.

\begin{proposition}\label{Coreleastmajorized}
In discrete convex games $\Gamma^{\rm c}_{\mathbf{Z}}$, 
the core admits a least majorized element.
\qed\end{proposition}

\subsection{Properties of the egalitarian solution in Case $\mathbf{Z}$}
\label{PropertyofEgaZ}
We first consider the Property 1. 
Example \ref{ExampleofMultipleEga} below shows that 
there can exist multiple egalitarian solutions in Case $\mathbf{Z}$. 
That is, the Property 1 does not hold in Case $\mathbf{Z}$.

\begin{example}\label{ExampleofMultipleEga}
\rm
Let $N = \{1,2,3\}$, 
$v(\{i\}) = 0 ~ (i \in N)$, 
$v(\{1,2\}) = v(\{1,3\}) = v(N) = 1$, and
$v(\{2,3\}) = 0$. 
The egalitarian solutions are 
$E(L(N,v)) = \{(1,0), (0,1)\}$, 
which implies the non-uniqueness of the egalitarian solution
in Case $\mathbf{Z}$.
\qed\end{example}
 
Next, we consider the Property 2. 
The non-uniqueness of the egalitarian solution in discrete games
suggests two separate problems in Case $\mathbf{Z}$. 
The first question is 
whether there exists at least one egalitarian solution 
for any discrete convex game. 
The second is what is the relationship 
between the core and the egalitarian solution.\par
 	
We first consider the existence of the egalitarian solution 
in discrete convex games. 
The following fundamental result is a key property in this paper.

\begin{theorem}\label{Existenceofleastmajorizedelement}
For any discrete convex game $(N,v) \in \Gamma^{\rm c}_{\mathbf{Z}}$, 
there exists some $x \in C(N,v)$ that Lorenz-dominates any $y \in C(N,v)$ 
with $y{\downarrow} \neq x{\downarrow}$.
\qed\end{theorem}

Theorem \ref{Existenceofleastmajorizedelement} 
is derived from Proposition \ref{Coreleastmajorized} as follows. 
Proposition \ref{Coreleastmajorized} and the fact that the core 
in discrete convex games 
is an {\rm M}-convex set
imply that there exists a least majorized element 
in the core. 
By these facts and the relationship between being majorized 
and Lorenz-domination as seen in Section \ref{Sub3:Def}, 
we obtain that for any discrete convex game, 
there exists an element of the core 
that Lorenz-dominates every element of the core 
not value-equivalent to the element. 
\par

Next, we show the existence of 
the egalitarian solution in any discrete convex game. 
Note that the following theorem does not state 
that all egalitarian solutions are contained in the core.

\begin{theorem}\label{ExistenceofEgalitarian}
In discrete convex games, 
there exists an egalitarian solution, 
that is, for any $(N,v) \in \Gamma^{\rm c}_{\mathbf{Z}}$, 
$E(L(N,v)) \neq \emptyset$ holds.
\qed\end{theorem}
\begin{proof}
Assume, to the contrary, that $E(L(N,v)) = \emptyset$. 
Take any $x \in C(N,v)$ in Theorem \ref{Existenceofleastmajorizedelement}. 
If there is no element of the Lorenz core that Lorenz-dominates $x$, 
then $x \in E(L(N,v))$ holds, which contradicts the assumption 
that $E(L(N,v)) = \emptyset$. 
Therefore there exists some $y \in L(N,v)$ 
that Lorenz-dominates $x$. 
Here we obtain $y \notin C(N,v)$, 
since otherwise $x$ is Lorenz-dominated by the core element $y$, 
which contradicts the fact that $x$ is not Lorenz-dominated 
by any element of the core by Theorem \ref{Existenceofleastmajorizedelement}. 
Thus we have $y \in L(N,v) \setminus C(N,v)$. 
Note that $y \neq x$. 
\par
 		
Then, $E(L(N,v)) = \emptyset$ shows that $y \notin E(L(N,v))$. 
Hence there exists some $y_1 \in L(N,v)$ 
that Lorenz-dominates $y$. By the above argument, 
we have $y_1 \notin C(N,v)$ and $y_1 \neq y$. 
Moreover, $y_1 \neq x$ holds since if $y_1$ Lorenz-dominates $y$, 
then $y_1$ Lorenz-dominates $x$. 
By repeating the above arguments, 
we arrive at the Lorenz core element $y_k$ 
that is not Lorenz-dominated by any element of $L(N,v)$ 
because $L(N,v)$ is bounded. 
Note that $y, y_1, \dots, y_k, x$ are all distinct. 
However, this contradicts the assumption that $E(L(N,v)) = \emptyset$.
\end{proof}

\begin{remark}\label{Remark:externalLorenzStablity}
\rm
By the proof of Theorem \ref{ExistenceofEgalitarian}, 
we obtain that for any discrete convex game
$(N,v) \in \Gamma^{\rm c}_{\mathbf{Z}}$ and
for each $y \in L(N,v) \setminus E(L(N,v))$, 
there exists an $x \in E(L(N,v))$ 
that Lorenz-dominates $y$.
Llerena--Mauri \cite{Llerena=Mauri2017}
calls this property the {\em external Lorenz stability}.
\qed\end{remark}

Next, we consider the relationship 
between the core and the egalitarian solution in Case $\mathbf{Z}$. 
In Case $\mathbf{R}$, 
the egalitarian solution always belongs to the core 
for any convex game. 
The following example reveals that, 
in Case $\mathbf{Z}$, 
there can exist an egalitarian solution 
outside the core even in convex games. 
 
\begin{example}\label{CounterExample}
\rm
Let $N = \{1,2,3\}$ and define $v$ as in the following table. 
This example is based on Example 5 in Dutta--Ray \cite{Dutta-Ray1}.
 	
\begin{table}[htb]
\begin{center}
\caption{Values of $v$ and egalitarian solutions in Example \ref{CounterExample}}
\begin{tabular}{cccc} 
$S$ & $v(S)$ & $EL(S,v)$ \\ \hline
\{1\} & 40 &  \{40\}\\
\{2\} & 60 &  \{60\}   \\
\{3\} & 80 &  \{80\}    \\
\{1,2\} & 110 &\{(50,60)\}    \\
\{1,3\} & 120 & \{(40,80)\}\\
\{2,3\} & 150 & \{(70,80)\} \\
\{1,2,3\} & 210 & \{(60,70,80), (64,65,81), (65,64,81)\} \\ \hline
\end{tabular}
\end{center}
\end{table}
 	
This game is convex and $(60,70,80) \in E(L(N,v))$ 
is in the core. 
However, $(64,65,81)$ and $(65,64,81) \in E(L(N,v))$ 
are not contained in the core because for each vector, 
the sum of its second and third components is as follows respectively.
\[
\pushQED{\qed}
65 + 81 = 146 < 150 = v(\{2,3\}),\, 64+81=145 < 150 = v(\{2,3\}). \qedhere
\popQED
\]
\end{example}

Example \ref{CounterExample} poses the question
whether there always exists an egalitarian solution in the core. 
However, this question remains unsolved in this paper.
\par
 
As mentioned in Theorem \ref{Theorem:LexEga}, 
in convex games, 
lexicographically optimal base and 
the egalitarian solution are equivalent in Case $\mathbf{R}$. 
However, the fact that not all egalitarian solutions belong to the core 
even in discrete convex games shows that, 
the set of egalitarian solutions do not necessarily 
coincide with the set of dec-min elements 
in an M-convex set in Case $\mathbf{Z}$. 
\par

Finally,
we consider the Property 3. 
In discrete convex games, 
we have to consider two problems. 
(a) Whether every egalitarian solution Lorenz-dominates 
every element of the core 
that is not an egalitarian solution, 
(b) whether 
all egalitarian solutions satisfying (a) Lorenz-dominate
every other element of the Lorenz core.  
\par
 	
We first demonstrate that the egalitarian solutions 
outside the core do not satisfy the Property 3
and then we reveal that egalitarian solutions 
of the core have this property. 
\par

We consider Example \ref{CounterExample} again. 
The egalitarian solutions in the core of the game of Example \ref{CounterExample}
are $(64, 65, 81)$ and $(65,64, 81)$ (Table 1) 
and we take $(64,65,81) \in E(L(N,v))$. 
For  example, vector $(59, 71, 80)$ is in the core. 
Since its largest component 80 is smaller than 
that of $(64,65,81)$, $(64,65,81)$ does not Lorenz-dominate $(59, 71, 80)$. 
Note that $(59, 71, 80) \notin E(L(N,v))$ 
because $(60,70,80)$ Lorenz-dominates $(59, 71, 80)$. 
Thus, in discrete convex games, 
egalitarian solutions outside the core 
do not necessarily Lorenz-dominate every element of the core 
except for the egalitarian solution. 
This is the distinction of the case $\mathbf{R}$ and $\mathbf{Z}$. 
\par

We next consider the egalitarian solutions in the core. 
Since all least majorized elements in the core
Lorenz-dominate every element of the core 
that is not a least majorized element, 
the set of the egalitarian solutions in the core 
coincides with the set of the least majorized elements in the core. 
Therefore, in discrete convex games, 
the Property 3 holds for the egalitarian solutions in the core. 
\par
 	
In contrast, the Property 3 
does not hold for the elements of the Lorenz core as follows. 
For the game of Example \ref{CounterExample}, 
the egalitarian solution $(60,70,80) \in E(L(N,v))$ 
does not Lorenz-dominate $(64,64,82) \in L(N,v) \setminus E(L(N,v))$. 
Recall that $(60,70,80) \in C(N,v)$.
This fact shows that the egalitarian solutions of the core 
do not necessarily Lorenz-dominate every element of the Lorenz core 
that are not contained in $E(L(N,v))$.

\section{Reduced game property}
\label{Section:RGP}
By Example \ref{CounterExample}, 
we see
that there can exist an egalitarian solution 
outside the core even in discrete convex games. 
Therefore, in Case $\mathbf{Z}$, 
the egalitarian solution and the dec-min element of the core
are not equivalent. 
Also, we do not know about
the existence of the egalitarian solution of the core 
in discrete convex games. \par
	
Thus, 
we are motivated to consider the Lorenz stable set introduced 
by Arin--Inarra \cite{Arin_Inarra} and 
Hougaard et al. \cite{HougaardPelegPetersen2001}, 
a subset of the core consisting of the elements 
that are not Lorenz-dominated 
by any other element of the core. 
Their approach is based on the fact that 
the core is considered to be the set of natural stable allocations.
For example, in a class of cost and surplus sharing games, 
the core plays a crucial role
(see e.g., \cite{PelegSudholter}).
In particular, this class is contained in a class of games 
arising from combinatorial optimization problems
including the polymatroid theory,
where the core also plays a central role
(cf., \cite{Curiel1997,HougaardPelegPetersen2001,PelegSudholter}).
We follow their approach and 
show that the Lorenz stable set in discrete convex games 
has nice properties such as the Davis and Maschler reduced game property 
and the converse reduced game property 
in Sections \ref{Section:RGP} and \ref{Section:CRGP}. 
\par

Dutta \cite{Dutta} has already shown 
that the egalitarian solution in convex games in continuous variables 
has these nice properties. 
He derives these results 
by making use of the properties of the principal partition 
explained in Section \ref{Sub2:decompositonalgorithm}. 
The point in our study is that 
we can give the proofs of these properties
by utilizing the canonical partition and the canonical chain
due to Frank--Murota \cite{FK1,FK2}. 
\par

The results of Sections \ref{Section:RGP} and \ref{Section:CRGP}
are summarized as follows.

\begin{enumerate}
\item 
In discrete convex games, 
the Lorenz stable set is nonempty, 
and every element of it Lorenz-dominates every element of the core 
not contained in the Lorenz stable set 
(Theorems \ref{LASNotempty} and \ref{Theorem:LSALorenzdomination}).

\item
In discrete convex games, 
the Lorenz stable set has the Davis and Maschler reduced game property
and the converse reduced game property 
(Theorems \ref{LSADMRGPTheorem} and \ref{LASCRGPTheorem}).
\end{enumerate}

\subsection{Lorenz stable set}\label{Sub5:LSA}
Fisrt, we give the definition of the Lorenz stable set.

\begin{definition}\label{DefofLSA} 
\rm (Lorenz stable set \cite{Arin_Inarra,HougaardPelegPetersen2001})
For a game $(N,v) \in \Gamma_{\mathbf{R}}$, the {\em Lorenz stable set} 
${\rm LSS}(N,v)$ is defined as follows: 
\begin{align}
{\rm LSS}(N,v) = \{x \in C(N,v) \mid \nexists \, y \in C(N,v): y ~ \mbox{Lorenz-dominates} ~ x\}. \label{eq: Lorenzstableset}
\end{align}
\qed\end{definition}

We define the Lorenz stable set in discrete games 
by replacing $\Gamma_{\mathbf{R}}$
with $\Gamma_{\mathbf{Z}}$ in the above definition.

\begin{remark}\label{LorenzmaximalImputation}
\rm
Hougaard et al. (2001) \cite{HougaardPelegPetersen2001} 
introduced the notion of the {\em Lorenz maximal imputation}, 
whose definition is exactly same as the Lorenz stable set. 
In this paper, we use ``Lorenz stable set'' 
following Arin--Inarra (2001) \cite{Arin_Inarra}.
\qed\end{remark}

The Lorenz stable set is contained in the core by its definition. 
That is,
\begin{equation}
{\rm LSS}(N,v) \subseteq C(N,v) \label{eq: inclusion}
\end{equation}
\noindent
holds for each game $(N,v) \in \Gamma$
regardless of the case $\mathbf{R}$ and $\mathbf{Z}$.
\par

The following properties hold for the Lorenz stable set.

\begin{theorem} \label{LASNotempty}
For any discrete convex game, the Lorenz stable set is nonempty. 
\qed\end{theorem}

\begin{theorem} \label{Theorem:LSALorenzdomination}
For any discrete convex game $(N,v) \in \Gamma^{\rm c}_{\mathbf{Z}}$, 
if $x \in {\rm LSS}(N,v)$, 
then $x$ Lorenz-dominates every element of the core 
except for the elements value-equivalent to $x$.
\qed\end{theorem}

Theorem \ref{Theorem:LSALorenzdomination} shows that the Lorenz stable set 
has the Property 3. 
This is one of the reasons that 
we consider the Lorenz stable set 
instead of the egalitarian solution in Case $\mathbf{Z}$.

\subsection{Davis and Maschler reduced game property}\label{Sub5:DMRGP}
In this subsection, 
we consider
the Davis and Maschler reduced game property
of the Lorenz stable set
by using the properties of the canonical chain and the canonical partition 
describing the structures of dec-min elements of an M-convex set. 
\par
We first show that for every discrete convex game, 
the Lorenz stable set coincides with the set of dec-min elements of the core. 
This enables us to apply the results of Frank--Murota \cite{FK1,FK2}
to the study of the Lorenz stable set.

\begin{proposition} \label{LSAdecminEquivaProp}
In discrete convex games, 
the Lorenz stable set coincides with the set of dec-min elements in the core.
\qed\end{proposition}
\begin{proof}
Note first that since every least majorized element of the core
Lorenz-dominates every element of the core 
that is not a least majorized element, 
the Lorenz stable set coincides with 
the set of the least majorized elements of the core 
(see also Definition \ref{DefofLSA}). 
The existence of a least majorized element of the core is guaranteed by Proposition \ref{Coreleastmajorized}.
Since if the core admits a least majorized element, 
then an element is the least majorized in the core 
if and only if 
it is a dec-min element in the core by Proposition \ref{PropofTamir}, 
the Lorenz stable set coincides with the set of dec-min elements of the core. 
\end{proof}

\begin{remark}\label{equivaofIncmaxDecmin}
\rm
As noted in Section \ref{Subsec2Def}, 
even when we define the notion of Lorenz-domination in an increasing order, 
its change does not affect the results of
Sections \ref{Section:RGP} and \ref{Section:CRGP}. 
This is justified by the following property
(e.g., Frank--Murota \cite{FK2} and Tamir \cite{Tamir}). 
Let $Q$ be an arbitrary subset of $\mathbf{R}^N $ 
and assume that $Q$ admits a least majorized element. 
For any $x \in Q$ the following three conditions are equivalent.

\begin{enumerate}[(A)]
\item $x$ is least majorized in $Q$.
\item $x$ is decreasingly minimal in $Q$.
\item $x$ is increasingly maximal in $Q$.
\end{enumerate}
\noindent
Since, in discrete convex games, 
the Lorenz stable set coincides with the set of dec-min elements of the core 
by Proposition \ref{LSAdecminEquivaProp}, 
the equivalence between (B) and (C) implies that 
the Lorenz stable set coincides with the set of inc-max elements of the core. 
\qed\end{remark}

Next we define the reduced game 
and the Davis and Maschler reduced game property.

\begin{definition} \label{DefofReducedgame} \rm (Reduced game (Davis--Maschler \cite{DM}))
Let $(N,v) \in \Gamma_{\mathbf{R}}$ be a game, $S \subsetneq N$, 
and $x \in \mathbf{R}^N$ be a payoff vector. 
The {\em reduced game} with respect to $S$ 
and $x$ is the game $(S, v^x_S)$ where 
\begin{equation}
v^x_S(T) = \begin{cases}
0 & (T = \emptyset), \\
v(N) - x(N \setminus S)  & (T = S), \\
\max_{Q \subseteq N \setminus S} \{v(T \cup Q) - x(Q) \} & ( T \subsetneq S). \label{eq:reducedgame}
\end{cases}
\end{equation}
\qed\end{definition}

\begin{definition} \label{DefofDMRGP} \rm (Davis and Maschler reduced game property \cite{DM})
Let $\sigma$ be a solution over a class $\Gamma_{\mathbf{R}}$ of games. 
Then $\sigma$ is said to have the {\em Davis and Maschler reduced game property} 
over $\Gamma_{\mathbf{R}}$, 
when for all $(N,v) \in \Gamma_{\mathbf{R}}$, 
for all $x \in \sigma(N,v)$, and
for all $S \subsetneq N$, 
$(S, v^x_S) \in \Gamma_{\mathbf{R}}$ and $x_S \in \sigma(S, v^x_S)$ hold.
\qed\end{definition}

We define the reduced game and 
the Davis and Maschler reduced game property in Case $\mathbf{Z}$ 
by replacing $\mathbf{R}$ with $\mathbf{Z}$ in the above definitions. 
\par

Dutta \cite{Dutta} shows the following fact in Case $\mathbf{R}$. 
This also holds for any discrete convex game. 
For completeness, we give the proof.

\begin{lemma}\label{ReducedgameConvexLemma}
For any discrete convex game $(N,v) \in \Gamma^{\rm c}_{\mathbf{Z}}$, 
for all $S \subseteq N$, and 
for all $y \in {\rm LSS}(N, v)$, 
$(S, v^y_S)$ is a discrete convex game.
\qed\end{lemma}
\begin{proof}
For any $T_i \subseteq S ~ (i = 1,2)$, 
there exists some $R_i \subseteq N \setminus S $ 
such that
\begin{align}
v^y_S(T_i) = \max\{v(T_i \cup R) - y(R) \mid R \subseteq N \setminus S\} 
= v(T_i \cup R_i) - y(R_i). \nonumber
\end{align}
\noindent
By using the supermodularity of $v$, we have
\begin{align*}
&v^y_S(T_1) + v^y_S(T_2)   \\
&= v(T_1 \cup R_1) - y(R_1) + v(T_2 \cup R_2) - y(R_2)  \\
& = v(T_1 \cup R_1) + v(T_2 \cup R_2) - y(R_1 \cup R_2) - y(R_1 \cap R_2) \\
& \leq  v((T_1 \cup R_1) \cup (T_2 \cup R_2)) + v((T_1 \cup R_1) \cap (T_2 \cup R_2)) - y(R_1 \cup R_2) - y(R_1 \cap R_2) \nonumber \\
& =  v((T_1 \cup T_2) \cup (R_1 \cup R_2)) - y(R_1 \cup R_2) + v((T_1 \cap T_2) \cup (R_1 \cap R_2)) - y(R_1 \cap R_2) \\
& \leq  \max\{v((T_1 \cup T_2) \cup Q) - y(Q) \mid Q \subseteq N \setminus S\} \\
& \quad + \max\{v((T_1 \cap T_2) \cup Q) - y(Q) \mid Q \subseteq N \setminus S\} \\
& =  v^y_S(T_1 \cup T_2) + v^y_S(T_1 \cap T_2), 
\end{align*}
\noindent which shows the supermodularity of $v^y_S$.  
\end{proof}

Peleg \cite{Peleg} has already shown 
the Davis and Maschler reduced game property 
of the core in Case $\mathbf{R}$. 
This also holds for any discrete game. 
Its proof is exactly same as that of Peleg,
but,
we give the proof for the sake of completeness.
Note that we do not assume 
the convexity of games in the following theorem.

\begin{theorem}\label{Theorem:CoreDMRGP}  
The core has the Davis and Maschler reduced game property
for any discrete game. 
That is, 
for all $(N,v) \in \Gamma_{\mathbf{Z}}$, 
for all $x \in {\rm LSS}(N,v)$, 
and for all $S \subsetneq N$, 
$(S, v^x_S) \in \Gamma_{\mathbf{Z}}$ and
$x_S \in {\rm LSS}(S, v^x_S)$ hold. 
\qed\end{theorem}
\begin{proof}
Take any $x \in C(N,v)$ and 
$S \subseteq N ~ (S \neq \emptyset)$. 
We want to show that $x_S \in C(S, v^x_S)$. 
First we note that $x(S) = v^x_S(S)$. 
Indeed, if $T = S$, then using $x(N) = v(N)$, 
we have $v^x_S(T) - x(T) = v(N) - x(N \setminus S) - x(S) = v(N) - x(N) = 0$, 
which shows that $x(S) = v^x_S(S)$. If $T \subsetneq S$, 
then the following inequality holds.
\begin{align*}
v^x_S(T) -x(T) &  = \max \{v(T \cup Q) - x(Q) \mid Q \subseteq N \setminus S\} - x(T) \\
& =  \max \{v(T \cup Q) - x(T \cup Q) \mid Q \subseteq N \setminus S\} \\
&\leq 0. 
\end{align*} 
\noindent
This inequality implies that 
$x(T) \geq v^x_S(T)$
for any $T \subsetneq S$. 
Therefore we obtain $x_S \in C(S, v^x_S)$.
\end{proof}

Here we show 
the Davis and Maschler reduced game property 
of the Lorenz stable set
in discrete convex games. 
We emphasize that the proof of the following theorem 
relies heavily on the properties of the canonical chain and the canonical partition.

\begin{theorem} \label{LSADMRGPTheorem}
For any discrete convex game $(N,v) \in \Gamma^{\rm c}_{\mathbf{Z}}$, 
the Lorenz stable set has the Davis and Maschler reduced game property. 
That is, if $x \in {\rm LSS}(N,v)$, 
then $x_S \in {\rm LSS}(S,v^x_S)$ holds for all $S \subseteq N ~ (S \neq \emptyset)$.
\qed\end{theorem}
\begin{proof}
Assume, to the contrary, that for some $y \in {\rm LSS}(N,v)$ 
and for some  $T \subsetneq N ~ (T \neq \emptyset)$, 
$y_T$ is Lorenz-dominated by some $x \in {\rm LSS}(T, v^y_T)$. 
Note that $x \in C(T, v^y_T)$ by (\ref{eq: inclusion}). 
Then, we can prove the following claim, which is proved later.
\begin{claim}\label{Claim}
\begin{equation}
\exists k \in \{1, \dots, q\}: \sum_{i \in T \cap C_k} y_i > \sum_{i \in T \cap C_k} x_i, \label{eq: claim}
\end{equation}
\end{claim}
\noindent
where $\{C_1, \dots, C_q\}$ is the canonical chain for $N$ 
constructed by the iterative procedure in Section \ref{Sub3:Mconvex}. 
\par
For $k$ in Claim \ref{Claim}, 
let $R = T \cap C_k$ and $R_k = C_k \setminus T$. 
Then, we have
\begin{align}
v(R_k \cup R) & =  v(C_k) = \sum_{i \in C_k} y_i. \label{eq: tight}
\end{align}
\noindent
The second equality follows from 
$y(C_k) = v(C_k)$ for each $k = 1, \dots, q$ 
(cf., Theorem \ref{Theorem:structureofDecmin}). 
By the definition of $v^y_T$, (\ref{eq: tight}) 
and the inequality of Claim \ref{Claim}, 
we obtain
\begin{align}
v^y_T(R)  & = \max\{v(Q \cup R) - y(Q) \mid Q \subseteq N \setminus T\} \geq  v(R \cup R_k) - \sum_{i \in R_k} y_i \nonumber \\
& = \sum_{i \in R} y_i > \sum_{i \in R} x_i, \nonumber
\end{align}
\noindent
which contradicts $x \in {\rm LSS}(T, v^y_T) \subseteq C(T, v^y_T)$ (see also (\ref{eq: inclusion})). 
	
We now prove Claim \ref{Claim}. 
Assume, to the contrary, that
\begin{align}
\forall k \in \{1, \dots, q\} : \sum_{i \in T \cap C_k} y_i \leq \sum_{i \in T \cap C_k} x_i. \label{eq:notclaim}
\end{align}
\noindent
Under this assumption, 
we will show the value-equivalence of $x$ and $y_T$ on $T$, 
which contradicts the assumption that 
$x$ Lorenz-dominates $y_T$. 
Then we are done. 
\par
		
First, we show that $x$ and $y_T$ are value-equivalent on $T \cap C_1$. 
We may assume that $T \cap C_1 \neq \emptyset$. 
By Theorem \ref{Theorem:structureofDecmin} and (\ref{eq:near-uniform}), 
$y_j = \beta_1$ or $y_j = \beta_1 - 1$ holds for all $j \in T \cap C_1$. 
Two cases are to be distinguished.

\begin{enumerate}[(1)]
\item 
The case where $y_j = \beta_1$ 
for all $j \in T \cap C_1$. 
Since $x$ Lorenz-dominates $y_T$, 
we have $x_j \leq \beta_1$ 
for all $j \in T \cap C_1$. 
Then, this fact and (\ref{eq:notclaim}) show that 
$x_j = \beta_1$ holds for all $j \in T \cap C_1$. 
Therefore, $x$ and $y_T$ are value-equivalent on $T \cap C_1$.

\item 
The case where $y_j = \beta_1 -1$ 
for some $j \in T \cap C_1$. 
We show that
\begin{align}
\sum_{i \in T \cap C_1} y_i = \sum_{i \in T \cap C_1} x_i. \label{eq:ineq}
\end{align}
\noindent
Assume that $\sum_{i \in T \cap C_1} y_i < \sum_{i \in T \cap C_1} x_i$ holds. 
Then, since $x$ Lorenz-dominates $y_T$, 
$x_j \leq \beta_1$ holds for all $j \in T \cap C_1$. 
Hence, this inequality implies that 
the number of $\beta_1$-valued components of $y_T$ is strictly smaller 
than that of $x$ (see Figure \ref{Figure:Case1}), 
which contradicts the assumption that $x$ Lorenz-dominates $y_T$. 
Therefore, we have (\ref{eq:ineq}).
This equation, together with the facts that 
$x_i \leq \beta_1$ holds for all $i \in T \cap C_1$ 
and either $y_i = \beta_1$ or $y_i = \beta_1 -1$ holds,
shows that $x$ and $y_T$ are value-equivalent on $T \cap C_1$.
\end{enumerate}

\begin{figure}
\centering
\begin{equation}
\begin{aligned}
x :& \overbrace{\beta_1  = \dots = \beta_1 = \beta_1}^{r} > \overbrace{\beta_1 -1 = \dots = \beta_1 -1}^{|T \cap C_1| - r }  \nonumber  \\
y : & \overbrace{\beta_1 = \dots = \beta_1}^{{\rm at ~ most } ~{r-1} }> \beta_1 - 1 \geq \dots
\end{aligned}
\end{equation}
\caption{Values of $x$ and $y$ on $T \cap C_1$}
\label{Figure:Case1}
\end{figure}

Since $x$ and $y$ are value-equivalent on $T \cap C_1$ as above, 
(\ref{eq:notclaim}) implies the following:
\begin{align}
\sum_{i \in T \cap (C_2 \setminus C_1)} y_i \leq \sum_{i \in T \cap (C_2 \setminus C_1)} x_i. \label{eq:C_2 - C_1}
\end{align}
	
Next, we show the value-equivalence between
$x$ and $y_T$ on $T \cap C_2$, 
that is, $x$ and $y_T$ are value-equivalent on $T \cap (C_2 \setminus C_1)$. 
Note first that either $y_j = \beta_2$ or $y_j = \beta_2 - 1$ holds
for all $j \in T \cap (C_2 \setminus C_1)$
by Theorem \ref{Theorem:structureofDecmin} and (\ref{eq:near-uniform}).

\begin{enumerate}[(1)]
\item 
The case where 
$y_j = \beta_2$ for all $j \in T \cap (C_2 \setminus C_1)$. 
Then, since $x$ and $y_T$ are value-equivalent on $T \cap C_1$ 
and $x$ Lorenz-dominates $y_T$, 
we obtain  $x_j \leq \beta_2 \ (\forall j \in T \cap (C_2 \setminus C_1))$. 
This statement and (\ref{eq:C_2 - C_1}) show that 
$x_j = \beta_2$ holds for all $j \in T \cap (C_2 \setminus C_1)$, 
which implies that $x$ and $y_T$ 
are value-equivalent on $T \cap (C_2 \setminus C_1)$.
		
\item 
The case 
where $y_k = \beta_2 - 1$ for some $k \in T \cap (C_2 \setminus C_1)$. 
We will first show that $\sum_{i \in T \cap (C_2 \setminus C_1)} y_i = \sum_{i \in T \cap (C_2 \setminus C_1)} x_i$ holds. 
Assume that 
			
\begin{equation}
\sum_{i \in T \cap (C_2 \setminus C_1)} y_i < \sum_{i \in T \cap (C_2 \setminus C_1)} x_i.  \label{eq:case2}
\end{equation}

By the facts that $y_j$ is either $\beta_2$ or $\beta_2 - 1$ 
for any $j \in T \cap (C_2 \setminus C_1)$, 
$x$ Lorenz-dominates $y_T$, 
and $x$ and $y_T$ are value-equivalent on $T \cap C_1$, 
we have $x_j \leq \beta_2$ for all $j \in T \cap (C_2 \setminus C_1)$. 
Therefore, if (\ref{eq:case2}) is true, 
then the number of {$\beta_2$}-valued components of $y_T$ 
is strictly smaller than that of $x$ on $T \cap (C_2 \setminus C_1)$
(see Figure \ref{Figure:Case2}), 
which contradicts the assumption that $x$ Lorenz-dominates $y_T$ 
together, 
since $x$ and $y_T$ are value-equivalent on $T \cap C_1$.
\end{enumerate}

\begin{figure}
\centering
\begin{equation}
\begin{aligned}
x : & \overbrace{\overbrace{\beta_1  = \dots = \beta_1}^{r_1} > \overbrace{\beta_1 -1 = \dots = \beta_1 -1}^{|T \cap C_1| - r_1} }^{T \cap C_1}  \geq \overbrace{\overbrace{ \beta_2 = \dots = \beta_2 = \beta_2 }^{r_2} > \beta_2 - 1 = \dots = \beta_2 -1}^{T \cap (C_2 \setminus C_1)} 
\nonumber  \\
y : & \overbrace{\overbrace{\beta_1 = \dots = \beta_1}^{r_1} > \overbrace{\beta_1 -1 = \dots = \beta_1 -1}^{|T \cap C_1| - r_1}}^{T \cap C_1} \geq \overbrace{\overbrace{\beta_2 = \dots = \beta_2}^{{\rm at~ most~} r_2-1}> \beta_2 -1 \geq \dots}^{T \cap (C_2 \setminus C_1)}
\end{aligned}
\end{equation}
\caption{Values of $x$ and $y$ on $T \cap C_2$}
\label{Figure:Case2}
\end{figure}
From the above arguments, 
we obtain that $\sum_{i \in T \cap C_2} y_i = \sum_{i \in T \cap C_2} x_i$ 
and $x$ and $y_T$ are value-equivalent on $T \cap C_2$. 
By repeating this argument until $k = q$, 
we have that $x$ and $y_T$ are value-equivalent on $T$.
Thus the proof of Claim 1 is completed.
\end{proof}

\begin{remark}\label{CounterofEgaDM}
\rm
Here, we demonstrate that the egalitarian solution in discrete variables 
fails to have the Davis and Maschler reduced game property 
even in discrete convex games. 
We reconsider the game of Example \ref{CounterExample}. 
The value of $v^x_S$ and 
the set of egalitarian solutions of the reduced game 
with respect to $S = \{2,3\}$ and $x = (64,65,81) \in E(L(N,v))$ 
are given as in Table \ref{TableofRemark5.3}.
	
\begin{table}[htb]
\caption{Values of $v^x_S$ and the egalitarian solution of $(S, v^x_S)$}
\begin{center}
\begin{tabular}{cccc} 
$T$ & $v^x_S(T)$ & $EL(T,v^x_S)$ \\ \hline
\{2\} & 60 &  \{60\}   \\
\{3\} & 80 &  \{80\}    \\
\{2,3\} & 146 & \{(66,80)\} \\ \hline
\end{tabular}
\end{center}
\label{TableofRemark5.3}
\end{table}
	
The values of $v^x_S$ are calculated as follows:
\begin{align*}
&v^x_S(\{2\}) = \max\{v(\{2\}), v(\{1,2\}) - x_1\} = \max\{60, 110 - 64\} = 60,
\\ 
&v^x_S(\{3\}) = \max\{v(\{3\}), v(\{1,3\}) - x_1\} = \max\{80, 120 - 64\} = 80,
\\
&v^x_S(\{2,3\}) = v(\{1,2,3\}) - x_1 = 210 - 64 = 146. 
\end{align*}
\noindent
Then, we obtain $x_S = (65, 81) \notin EL(S, v^x_S)$, 
which shows that 
the egalitarian solutions outside the core 
fail to have the Davis and Maschler reduced game property 
even in discrete convex games. 
\qed\end{remark}
\section{Converse reduced game property}
\label{Section:CRGP}
In this section, we consider 
the converse reduced game property of the Lorenz stable set
in discrete convex games. 
Peleg \cite{Peleg} defines 
the converse reduced game property as follows.

\begin{definition} \label{DefofCRGP} \rm (Converse reduced game property (Peleg \cite{Peleg}))
Let $\sigma$ be a solution on $\Gamma_{\mathbf{R}}$. 
A solution $\sigma$ is said to have the {\em converse reduced game property} 
if the following condition is satisfied:
For $x \in \mathbf{R}^N$ with $x(N) = v(N)$, 
if $(N,v) \in \Gamma_{\mathbf{R}}$ and $x_S \in \sigma(S, v^x_S)$ 
for every $S \subseteq N$ with $|S| = 2$, 
then $x \in \sigma(N,v)$ holds.
\qed\end{definition}

We define the converse reduced game property 
in Case $\mathbf{Z}$ by replacing $\mathbf{R}$ with $\mathbf{Z}$ 
in the above definition.
\par

Peleg \cite{Peleg} has already shown 
the converse reduced game property 
of the core.
This is also true for any discrete game. 
Its proof is exactly same as that of Peleg, 
but, 
we give the proof
for the sake of completeness.
Note that we do not assume 
the convexity of games in the following lemma.

\begin{lemma}\label{CoreCRGPLemma}
For any discrete game $(N,v) \in \Gamma_{\mathbf{Z}}$, 
the core satisfies the converse reduced game property.
\qed\end{lemma}
\begin{proof}
For a discrete game $(N,v) \in \Gamma_{\mathbf{Z}}$, 
let $x \in \mathbf{Z}^N$ be a vector satisfying 
$x(N) = v(N)$ and $x_S \in \sigma(S, v^x_S)$ 
for every $S$ with $|S| = 2$. 
Take any $T \subsetneq N ~ (T \neq \emptyset)$, 
$i \in T$, and $j \in N \setminus T$. Let $Q = \{i, j\}$. 
Then, by using $x_Q \in C(Q, v^x_Q)$, 
we obtain the following inequalities:
\begin{eqnarray*}
0 \geq v^x_Q(\{i\}) - x_i & \geq & v((T \setminus \{i, j\}) \cup \{i\}) 
- x(T \setminus \{i, j\}) - x_i  \\
& = & v(T) - x(T),
\end{eqnarray*}
\noindent
where the equality is due to $j \notin T$.
Therefore, $x(T) \geq v(T)$ holds for 
every $T \subsetneq N$. Also, we have $x(N) = v(N)$ 
by the hypothesis of $x$. Thus, we obtain $x \in C(N,v)$. 
\end{proof}

\subsection{Converse reduced game property in Case $\mathbf{Z}$}
\label{Sub6:CRGPinZ}
In this subsection, 
we prove the converse reduced game property 
of the Lorenz stable set
in Case $\mathbf{Z}$. 

\begin{theorem}\label{LASCRGPTheorem}
For any discrete convex game $(N,v) \in \Gamma^{\rm c}_{\mathbf{Z}}$,
the Lorenz stable set has the converse reduced game property, 
that is, for $x \in \mathbf{Z}^N$ with $x(N) = v(N)$, 
if $(N,v) \in \Gamma^{\rm c}_{\mathbf{Z}}$ and $x_S \in {\rm LSS}(S, v^x_S)$ 
for every $S \subseteq N$ with $|S| = 2$,  then $x \in {\rm LSS}(N,v)$ holds.
\qed\end{theorem}
\begin{proof}
Note first that $x \in C(N,v)$ by Lemma \ref{CoreCRGPLemma}. 
Suppose that $x \notin {\rm LSS}(N,v)$. 
Since the Lorenz stable set coincides with 
the set of dec-min elements of $C(N,v)$ 
by Proposition \ref{LSAdecminEquivaProp}, 
$x$ is not a dec-min element of $C(N,v)$. 
Therefore, by Theorem \ref{1tighteningDecminTheorem}, 
there is a 1-tightening step for $x$, 
that is, there exist some $i,j \in N$ such that 
$x_j \geq x_i + 2$ and $x' = x + \chi_i - \chi_j \in C(N,v)$. 
Since the core satisfies the Davis--Maschler reduced game property 
by Theorem \ref{Theorem:CoreDMRGP}, we have $x'_{\{i,j\}} \in C(\{i,j\}, v^{x'}_{\{i,j\}})$. 
Note that $x_{\{i,j\}} \in {\rm LSS}(\{i,j\},v^x_{\{i,j\}}) \subseteq C(\{i,j\}, v^x_{\{i,j\}})$ by the hypothesis. 
\par
		
Let $R = \{i,j\}$. Then, 
using the equation $x_{N \setminus R} = x'_{N \setminus R}$, 
we can show that 
\begin{align}
v^{x'}_R = v^x_R. \label{eq:exchange}
\end{align}
\noindent
Indeed, for any $T \subsetneq R ~ (T \neq \emptyset)$, we obtain
\begin{align*}
v^{x'}_R(T) &= \max\{v(T \cup Q) - x'(Q) \mid Q \subseteq N 
\setminus R\} \\
& = \max\{v(T \cup Q) - x(Q) \mid Q \subseteq N \setminus R\} \\
& = v^x_R(T),
\end{align*}
\noindent
where the first equality is due to (\ref{eq:reducedgame})
and the second equality follows from
$ x_{N \setminus R} = x'_{N \setminus R}$.
Also, if $T = \emptyset$, then we have 
$v^{x'}_{R}(\emptyset) = v^x_R(\emptyset) = 0$ 
by the definition of $v^{x'}_{R}$
(see Definition \ref{DefofReducedgame}). 
Similarly, for $T = R$, 
we can show $v^{x'}_R (R)= v^x_R(R)$ as follows:
\begin{align}
v^{x'}_R(R) &= v(N) - x'(N \setminus R) \nonumber \\
& = v(N) - x(N \setminus R) \nonumber \\
& = v^x_{R}(R). \nonumber
\end{align}
\noindent
From the above arguments we obtain (\ref{eq:exchange}). Hence we have 
\begin{align}
x'_R \in C(R, v^{x'}_R) = C(R, v^x_R).
\end{align}
\noindent
It follows from $x_j \geq x_i + 2$ that $x'_R= (x_i +1,x_j -1)$ 
Lorenz-dominates $x_R= (x_i,x_j)$, which contradicts $x_R \in {\rm LSS}(R,v^x_R)$ 
(see also Definition \ref{DefofLSA}).
\end{proof}

\section{Conclusion}
In this paper, we have pointed out that the egalitarian solution 
does not have nice properties in games with discrete side payment. 
Then, we have focused on the Lorenz stable set 
and shown that it has nice properties 
such as the Davis and Maschler reduced game property 
and the converse reduced game property.
The existence of the egalitarian solution of the core 
in discrete convex games is left for the future.

\section{Acknowledgment}
The author is grateful to Kazuo Murota 
for suggesting this research and for encouragement. 
He also thanks Takuya Iimura and 
Takahiro Watanabe for helpful comments.

\end{document}